\documentclass[11pt]{article}
\usepackage[margin=1in]{geometry}

\usepackage{times,mathptmx}
\usepackage{amsmath,amsthm,amssymb,xspace,verbatim,paralist,enumitem,multirow,bm,bbm,stmaryrd,xfrac}
\usepackage{mleftright,array,multirow,booktabs,multicol,relsize}
\usepackage[dvipsnames]{xcolor}
\usepackage{graphicx}
\usepackage{algorithm}
\counterwithin{algorithm}{section}
\usepackage[noend]{algpseudocode}
\usepackage{url,hyperref,cleveref}

\allowdisplaybreaks[1]

\newtheorem{theorem}{Theorem}[section]
\newtheorem{lemma}[theorem]{Lemma}

\newtheorem{conj}[theorem]{Conjecture}

\newtheorem{claim}[theorem]{Claim}

\newtheorem{proposition}[theorem]{Proposition}
\theoremstyle{definition}  \newtheorem{definition}[theorem]{Definition}

\theoremstyle{remark}

\newcommand{\RR}{\mathbb{R}}
\newcommand{\ZZ}{\mathbb{Z}}

\newcommand{\cF}{\mathcal{F}}

\newcommand{\cM}{\mathcal{M}}
\newcommand{\cO}{\mathcal{O}}

\newcommand{\cW}{\mathcal{W}}

\newcommand{\cZ}{\mathcal{Z}}
\newcommand{\bzero}{\mathbf{0}}
\newcommand{\bone}{\mathbf{1}}
\newcommand{\bb}{\mathbf{b}}
\newcommand{\be}{\mathbf{e}}

\newcommand{\br}{\mathbf{r}}
\newcommand{\btr}{\mathbf{\tilde{r}}}
\newcommand{\bs}{\mathbf{s}}
\newcommand{\bts}{\mathbf{\tilde{s}}}

\newcommand{\bu}{\mathbf{u}}
\newcommand{\bv}{\mathbf{v}}
\newcommand{\bw}{\mathbf{w}}
\newcommand{\bx}{\mathbf{x}}
\newcommand{\by}{\mathbf{y}}
\newcommand{\bz}{\mathbf{z}}

\newcommand{\tS}{\widetilde{S}}

\newcommand{\tOmega}{\widetilde{\Omega}}

\DeclareMathOperator{\cost}{cost}

\DeclareMathOperator{\LQ}{LQ}
\DeclareMathOperator{\poly}{poly}

\newcommand{\ceil}[1]{{\lceil{#1}\rceil}}

\renewcommand{\b}{\{0,1\}}
\newcommand{\txp}{\mathsf{T}}
\newcommand{\eps}{\varepsilon}
\newcommand{\setm}{\smallsetminus}
\newcommand{\ur}{\textsc{ur}\xspace}
\newcommand{\urs}{\textsc{ur}^\subset\xspace}

\newcommand{\elemx}{\textsc{elemx}\xspace} % single element extraction
\newcommand{\elemxodd}{\textsc{elemx}\ensuremath{^{(2)}}\xspace} % single element extraction, z == 1 (mod 2)
\newcommand{\elemxq}{\textsc{elemx}\ensuremath{^{(q)}}\xspace} % single element extraction, z != 0 (mod q)
\newcommand{\elemxqh}{\textsc{elemx}\ensuremath{^{(q,h)}}\xspace} % single element extraction, z == h (mod q)
\newcommand{\elemxquarter}{\textsc{elemx}\ensuremath{^{1/4}}\xspace} % single element extraction, |z|=n/4

\newcommand{\gen}[1]{{\langle{#1}\rangle}}

\makeatletter
\newcommand{\tpmod}[1]{{\@displayfalse\pmod{#1}}}
\makeatother

\newlist{thmparts}{enumerate}{1}
\setlist[thmparts]{labelindent=\parindent,leftmargin=*,itemsep=4pt,font=\normalfont,label=(\thetheorem.\arabic*)}
\Crefname{thmpartsi}{Property}{Property} % any better name? subresult? or Theorem?

\title{The Element Extraction Problem and the Cost of\\ Determinism and Limited Adaptivity in Linear Queries}

\author{Amit Chakrabarti\footnote{\url{amit.chakrabarti@dartmouth.edu}, Department of Computer Science, Dartmouth College} \and Manuel Stoeckl\footnote{\url{manuel.stoeckl.gr@dartmouth.edu}, Department of Computer Science, Dartmouth College}}

\date{}

\begin{document}

\maketitle

\thispagestyle{empty}
\begin{abstract}
Two widely-used computational paradigms for sublinear algorithms are using linear measurements to perform computations on a high dimensional input and using structured queries to access a massive input. Typically, algorithms in the former paradigm are non-adaptive whereas those in the latter are highly adaptive. This work studies the fundamental search problem of \textsc{element-extraction} in a query model that combines both: linear measurements with bounded adaptivity.

In the \textsc{element-extraction} problem, one is given a nonzero vector $\mathbf{z} = (z_1,\ldots,z_n) \in \{0,1\}^n$ and must report an index $i$ where $z_i = 1$. The input can be accessed using arbitrary linear functions of it with coefficients in some ring. This problem admits an efficient nonadaptive randomized solution (through the well known technique of $\ell_0$-sampling) and an efficient fully adaptive deterministic solution (through binary search). We prove that when confined to only $k$ rounds of adaptivity, a deterministic \textsc{element-extraction} algorithm must spend $\Omega(k (n^{1/k} -1))$ queries, when working in the ring of integers modulo some fixed $q$. This matches the corresponding upper bound. For queries using integer arithmetic, we prove a $2$-round $\tOmega(\sqrt{n})$ lower bound, also tight up to polylogarithmic factors. Our proofs reduce to classic problems in combinatorics, and take advantage of established results on the {\em zero-sum problem} as well as recent improvements to the {\em sunflower lemma}.


\bigskip
\paragraph{Acknowledgements:} We thank Deeparnab Chakrabarty, Prantar Ghosh, and anonymous reviewers for several helpful discussions, comments, and pointers.

\medskip\noindent
This work was supported in part by NSF under award 2006589.
\end{abstract}

\clearpage
\addtocounter{page}{-1}

%!TEX root = main.tex

\section{Introduction} \label{sec:intro}

Determinism versus randomization in algorithm design is a fundamental concern in computer science and is the topic of a great many works in complexity theory. In ``space-constrained'' models such as communication complexity and data streaming, basic results show that derandomization can entail an exponential or worse blow-up in cost. For instance, in the two-party communication setting, the very basic $n$-bit \textsc{equality} problem admits a bounded-error randomized protocol with only $O(1)$ communication ($O(\log n)$ if restricted to private coins), whereas its deterministic communication complexity is as large as it gets, namely $n+1$. In the data streaming setting, the similarly basic \textsc{distinct-elements} problem admits a one-pass bounded-error randomized algorithm that uses $O(\log n)$ space to provide a $(1+\eps)$-approximation~\cite{KaneNW10-pods}, whereas a deterministic algorithm would require $\Omega(n)$ space, even if multiple passes and large approximation factors are allowed~\cite{ChakrabartiK16}. In this work, we explore such a price-of-determinism phenomenon in the \emph{query complexity} world, for a similarly basic search problem.
%{\color{red} Our results, by their nature, hint at (as yet unproven) separations between determinism and randomization in communication complexity and in data streaming.}

The focus of our study is a search problem that we call \textsc{element-extraction} (henceforth, $\elemx$), where the input is a set $Z \subseteq [n] := \{1,\ldots,n\}$, promised to be nonempty, and the goal is to extract any element from $Z$. Formally, this is a total search problem given by the relation $\elemx_n \subseteq 2^{[n]} \times [n]$, where
\begin{align} \label{eq:elemx-def-set}
  \elemx_n = \left\{ (Z,i) : Z \subseteq [n],\, i \in [n],\, \text{and } |Z| > 0 \Rightarrow i \in Z \right\} \,.
\end{align}
As is often the case, the natural correspondence between sets in $2^{[n]}$ and vectors in $\b^n$ will be useful. Indeed, we shall freely switch between these two viewpoints, using the notational convention that uppercase letters denote sets and their corresponding lowercase boldface variants denote characteristic vectors. Thus, we can also formalize $\elemx$ as
\begin{align} \label{eq:elemx-def-vec}
  \elemx_n = \left\{ (\bz,i) : \bz = (z_1, \ldots, z_n) \in \b^n,\, i \in [n],\, \text{and } \bone^\txp \bz > 0 \Rightarrow z_i = 1 \right\} \,.
\end{align}
The goal of an algorithm solving $\elemx$ is to produce an output $i$ such that $(Z,i) \in \elemx$: with certainty in the deterministic setting, and with probability $\ge 2/3$ (say) in the randomized setting. In other words, the algorithm must produce a \emph{witness} of the nonemptiness of $Z$. To do so, the algorithm may access $Z$ (equivalently, $\bz$) using {\em linear queries}, as we shall now explain.

In a Boolean decision tree model, an algorithm may only access the input vector by querying its individual bits. In such a setting, there is not much to say about $\elemx$: even randomized algorithms are easily seen to require $\Omega(n)$ queries. But things get interesting if we allow more powerful queries: specifically, linear ones. Let us define a \emph{linear query protocol over domain $D$} (a $D$-LQP, for short) to be a query protocol wherein each query is an evaluation of a linear form $\sum_{i=1}^n a_i z_i$, where each $a_i \in D$. The domain $D$ should be thought of a ``reasonable'' subset of a ring containing $\b$---e.g., a finite field, or integers with bounded absolute value---and the linear functions will be evaluated in the underlying ring.  The cost of an LQP is the number of linear form evaluations used.\footnote{Note that this is somewhat lower than the number of {\em bits} needed to encode the output of the queries.}
In this work we particularly care about the amount of {\em adaptivity} in an LQP, which quantifies the extent to which each query depends on the outcomes of previous queries.

To set the stage, we recall the problem of $\ell_0$\textsc{-sampling}~\cite{FrahlingIS08,CormodeF14}, from the world of sketching and streaming algorithms. The goal of $\ell_0$-sampling is to sample a pair $(i, x_i)$ from a nonzero input vector $\bx \in \RR^n$ (say), so that $x_i \ne 0$ and $i$ is distributed nearly uniformly on the support of $\bx$. This is a fundamental primitive, used as a low-level subroutine in a wide range of applications in streaming and other ``big data'' algorithms. There are several solutions to this problem~\cite{CormodeF14}, most of which provide a linear {\em sketching} scheme, wherein one computes $\by = S\bx$ for a certain random $d \times n$ matrix $S$ and then runs a recovery algorithm on the low-dimensional vector $\by$ to produce the desired sample. Notice that if the input is a vector $\bz \in \b^n$, such a scheme provides a randomized LQP for $\elemx_n$ (allowing a small probability of error). In particular, using the optimal $\ell_0$-sampling sketch of Jowhari, Sağlam, and Tardos~\cite{JowhariST11}, we obtain a $\ZZ$-LQP that makes $O(\log n)$ queries, using coefficients in $\{0,1,\ldots,n\}$, and has the pleasing property of being {\em non-adaptive}. We can also obtain a $\ZZ_q$-LQP that makes $O(\log^2 n/\log q)$ queries;\footnote{Throughout this paper, ``$\log$'' denotes the base-$2$ logarithm.} details in \Cref{sec:upper-bounds}.

Turning to the deterministic setting---our main focus in this paper---it is easy to show that a non-adaptive $\ZZ$-LQP for $\elemx_n$ must make $\Omega(n/\log n)$ queries, for basic information-theoretic reasons. For completeness, we give the proof in \Cref{prop:int-one-round}. However, this heavy determinism penalty disappears upon moving to general deterministic LQPs, where we can use adaptivity. Indeed, a simple binary search strategy leads to a $\ZZ$-LQP that makes $O(\log n)$ queries, using coefficients in $\b$. We can refine this observation to trade off the query complexity for amount of adaptivity. This brings us to our central concept.

Define a {\em $k$-round LQP} to be one where the queries are made in batches that we call {\em rounds}: the collection of linear forms defining the queries in round $i$ depend only on the results of queries made in rounds $1,\ldots,i-1$ (a formal definition appears in \Cref{sec:prelim}). Then, a natural generalization of the binary search strategy provides a $k$-round $\ZZ$-LQP for $\elemx$, using coefficients in $\b$, making at most $k (\lceil n^{1/k} \rceil - 1)$ queries in total. When we are additionally promised that $\bone^\txp \bz \ne 0$, where addition is performed in the ring $\ZZ_q$, then this algorithm also works as a $\ZZ_q$-LQP; details in \Cref{sec:upper-bounds}. Notice that $k$-round LQPs naturally interpolate between linear sketches at one extreme (when $k = 1$) and linear decision trees at the other (when $k = n$).

The most important message of this paper is that the above rounds-versus-queries tradeoff is asymptotically tight for deterministic linear query protocols for $\elemx$, in several natural settings. We state our results informally for now, with formal statements given after the necessary definitions and preliminaries.

\subsection{Our Results and Techniques} \label{sec:results}

We shall study $D$-LQPs for the domains $D = \ZZ_q$, the ring of integers modulo $q$ (with $q \ll n$) as well as $D = \ZZ$, but with coefficients of small magnitude (at most $\poly(n)$, say). Such restrictions on the coefficients are necessary, because allowing arbitrary integer coefficients makes it possible to recover the entire input $\bz$ with the single query $\sum_{i=1}^n 2^{i-1} z_i$.

When $D = \ZZ_q$, for small $q$, solving $\elemx$ without the promise that $\bone^\txp \bz \ne 0$ is hard, regardless of the number of rounds. Intuitively, there is no cheap way to deterministically verify that a subset $I \subseteq [n]$ indeed contains an index $i \in I$ where $z_i \ne 0$. Defining the ``cost'' of an LQP to be the number of queries it makes in the worst case (formally defined in \Cref{sec:prelim}), we obtain the following not-too-hard results.

\begin{proposition} \label{thm:z2-hard}
Every deterministic $\ZZ_2$-LQP for $\elemx_n$ has cost $\ge n-1$, which is optimal.
\end{proposition}
\begin{proposition} \label{thm:zq-hard}
For $q\ge 3$, every deterministic $\ZZ_q$-LQP for $\elemx_n$ has cost $\ge {n}/(2 q \ln q)$.
\end{proposition}

As noted earlier, adding the promise that $\bone^\txp \bz \ne 0$ permits a more efficient $k$-round deterministic algorithm. For each integer $q \ge 2$, define $\elemxq_n$ to be the version of $\elemx_n$ where we are given the stronger promise that $\bone^\txp \bz \ne 0$ under arithmetic in $\ZZ_q$. Equivalently, using set notation, we are promised that $|Z| \not\equiv 0 \pmod q$. We prove the following results, using similar round-elimination arguments.

%%\mscomment{Alternative formats: $\overset{\left|z\right|\!\!\bmod\!2\ne0}{\textsc{elemx}}_{n}$; $\textsc{elemx}^{\oplus_2 }_{n}$ $\textsc{elemx}^{\text{odd}}_{n}$ $\textsc{elemx}^\oplus_2 (n)$}

\begin{theorem} \label{thm:z2-k-round}
Every deterministic $k$-round $\ZZ_2$-LQP for $\elemxodd_n$ has cost $\ge k (n^{1/k} - 1)$.
\end{theorem}
\begin{theorem} \label{thm:zq-k-round}
Every deterministic  $k$-round $\ZZ_q$-LQP for $\elemxq_n$ has cost $\Omega\left( \frac{1}{q^{1 + 1/k} \ln^2 q} k (n^{1/k} - 1) \right)$.
\end{theorem}
Although \Cref{thm:zq-k-round} subsumes \Cref{thm:z2-k-round} in the asymptotic sense, we find it useful to present the former result in full, first, to lay the groundwork for our subsequent lower bound proofs. As we shall see, the fact that $\ZZ_2$ is a field leads to an especially clean execution of the round elimination strategy. Note also that a weaker form of \Cref{thm:z2-k-round} follows from existing work on formula size-depth tradeoffs (see \Cref{sec:kw-lb}); however, the resulting proof, once fully unrolled, is considerably more complex than our direct argument.

At a high level, a lower bound proof based on round elimination works as follows. We consider a hypothetical $k$-round protocol for $n_k$-dimensional instances of some problem $P$ that does not incur much cost in its first round. Based on this low cost, we extract a $(k-1)$-round protocol for $n_{k-1}$-dimensional instances of $P$ by ``lifting'' these smaller instances to special $n_k$-dimensional instances on which the $k$-round protocol essentially ``wastes'' its first round. If we can carry out this argument while ensuring that the shrinkage from $n_k$ to $n_{k-1}$ is not too drastic, then a too-cheap $k$-round protocol will eventually give us a $0$-round protocol for a nontrivial instance dimension, leading to a contradiction.

In the proofs of the above two theorems, this strategy is executed by identifying a large collection of pairwise disjoint sets that are treated identically in the protocol's first round. Viewing these sets as blocks of indices within $[n]$, we consider block-structured instances of $\elemx_n$ and proceed to lift general instances of $\elemx_{n'}$ into these block-structured ones. In \Cref{thm:z2-k-round}, these blocks arise from elementary linear algebraic considerations. In \Cref{thm:zq-k-round}, the fact that inputs are in $\b^n$ instead of $\ZZ_q^n$ necessitates a brief excursion into additive combinatorics.

Finally, we consider LQPs over $\ZZ$, the ring of all integers, but with bounds on the magnitude of coefficients (which, as we noted earlier, is necessary in order to have nontrivial results). To be precise, we consider domains of the form $\ZZ_{[b,c]} := \{a \in \ZZ:\, b \le a \le c\}$. While we are unable to prove a full tradeoff lower bound in this case, we do obtain a near-optimal result for $k=2$ rounds.

\begin{proposition} \label{prop:int-one-round}
  Every deterministic $1$-round $\ZZ_{[-B,B]}$-LQP for $\elemx_n$
  costs $\Omega({n}/\log(nB))$.
\end{proposition}

\begin{theorem} \label{thm:int-lower-bounds} 
  Every deterministic $2$-round $\ZZ_{[-B,B]}$-LQP for $\elemx_n$
  costs $\Omega(\sqrt{n} / \log^{3/2}(nB) )$.
\end{theorem}

The former result is straightforward, based on the simple observation that such an LQP can extract the {\em entire} input $\bz$ followed by basic information theoretic considerations. Incidentally, the problem of extracting all of $\bz$ using $\ZZ_{[0,1]}$-LQPs has a long history as the {\em coin weighing} problem, for which a 1-round $O(n / \log n)$ algorithm exists; see \Cref{sec:related}.% I explicitly mentioned the upper bound here, since the text makes it clear it applies to elemx

The significant result here is the latter. It again uses a round elimination strategy and, as before, the bird's-eye view is that we identify disjoint blocks of indices to engineer a suitable lifting. This time, the blocks arise out of extremal combinatorics considerations, specifically the sunflower lemma, in its recently strengthened form~\cite{Rao20}. Furthermore, upon carrying out this round elimination, we are left with a $1$-round LQP that solves $\elemx$ only under a cardinality constraint on the input set. To finish the proof, we must demonstrate hardness even for this special case. This is not as straightforward as \Cref{prop:int-one-round}: our argument to handle this hinges on the Frankl--Wilson theorem~\cite{FranklW81} on set systems with forbidden intersection sizes.

Attempts to extend the above proof outline to handle more than two rounds runs into technical issues of integer divisibility. We suspect that this is an artifact of our proof machinery and not inherent to the problem. We conjecture that every deterministic $k$-round $\ZZ_{[-B,B]}$-LQP requires cost $\tOmega(n^{1/k})$, suppressing polylogarithmic factors. Indeed, we believe that much more is true, and that a communication complexity analogue of such a tradeoff also holds. We shall take this up after a discussion of related work.

\subsection{Related Work and Connections} \label{sec:related}

Our work touches upon several themes with long histories of study in computer science: determinism versus randomization, adaptivity versus non-adaptivity, sublinear algorithms, and input access through structured queries. With these connections in mind, we recall a small number of works that are either close in spirit to ours or shed light on some aspect of this work.

The most basic query model is the Boolean decision tree. In this setting, deterministic and randomized complexities are polynomially related for total Boolean functions~\cite{BuhrmanW02,AmbainisBBLSS17}, whereas arbitrarily large gaps are possible for search problems~\cite{LovaszNNW95}. Parity decision trees---equivalent to our $\ZZ_2$-LQPs---have been studied in several works (e.g., \cite{ZhangS10,HatamiHL18} and the references therein), usually for Boolean functions and focusing on connections with communication complexity of XOR-composed functions. Beyond the Boolean---or indeed the discrete---setting lie {\em linear decision trees}, where the input is a {\em real} vector and one can query the {\em sign} of a linear form~\cite{DobkinL78,KaneLM19}. All such ``decision tree'' models are fully adaptive and the vast majority of works using them do not focus on amount of adaptivity as a resource.

At the other extreme is the (nonadaptive) linear sketching model, where a high-dimensional input is accessed through one batch of linear queries (equivalently, through a low-dimensional sketch of it produced by a sketching matrix). This paradigm is ubiquitous in data streaming algorithms and compressed sensing~\cite{Muthukrishnan05,Donoho06,Woodruff14,CormodeYi-book} and has connections to dimension reduction and metric embeddings. Some recent work carries the message that linear sketching might be a {\em complete} paradigm for a large class of data streaming algorithms~\cite{LiNW14} and certain communication protocols~\cite{KannanMSY18,HosseiniLY19}. Most work on linear sketching considers {\em randomized} sketches, since determinism often precludes sublinear cost.

Turning to determinism, the well-studied {\em coin weighing} problem, put in our terms, asks for a $\ZZ_{[0,1]}$-LQP that retrieves the entire input $\bz \in \b^n$. It has long been known that $(2\pm o(1))n/\log n$ nonadaptive queries are necessary and sufficient. Special cases and variants of this problem have been studied over the years; see \cite{ErdosR63} for some early history and \cite{Bshouty09,MarcoK13} for recent history. While some of these works consider adaptive LQPs, there is no strong rounds-vs-queries tradeoff for this problem, which is harder than $\elemx$.

The body of work on round complexity under linear queries is much smaller. There is one recent work very close to ours: Assadi, Chakrabarty, and Khanna~\cite{AssadiCK20} studied a problem very similar to $\elemx$ that they called \textsc{single-element-recovery}, where the input is a vector $\bx \in \RR_{\ge0}^n$, and by applying $\RR$-linear queries one wishes to recovery an arbitrary element from the support of $\bx$. While their query model is much stronger than our $\ZZ$-linear or $\ZZ_q$-linear queries, it is balanced by the $\RR_{\ge 0}$-valued inputs that prevent tricks to recover the entire input in one query. Their main theorem implies that the deterministic $k$-round search algorithm making roughly $k (n^{1/k} - 1)$ queries in total---very similar to \Cref{alg:det-protocol}---has cost exactly matching the lower bound.
% see Theorem 3, for transcript creation; combined with our 'division sequences', gives result
% While their theorem is proven by using induction on the number of rounds to find an impossible protocol transcript, it can \emph{possibly} be converted to a round elimination argument, just as our own arguments can be modified to work with their proof strategy.
%
Linear queries and adaptivity are also featured together in some work on {\em sparse recovery} problems. One such problem is to find an approximately closest $s$-sparse vector $\bx^\star$ to an input $\bx$,
%%; specifically, one for which
%%\begin{align*}
%%  \left\Vert x - x^\star \right\Vert_2^2 \le C \min_{\text{$k$-sparse $x'$}} \left\Vert x - x' \right\Vert_2^2 \,.
%%\end{align*}
using $\RR$-linear queries to the input and $r$ rounds of adaptivity. For this, \cite{PriceW13} have proven near optimal lower bounds of $\Omega(r (\log n)^{1/r})$ when $s=1$ and \cite{KamathP19} have extended them to small $s$, proving $\Omega(\frac{1}{r} s (\log n)^{1/r})$ queries are needed when $\log s < (\log n)^{1/r}$.

\medskip
A number of works consider rounds of adaptivity in query models beyond linear queries. Recent examples include works on maximizing submodular functions through adaptive oracle queries~\cite{BalkanskiS18}; on adaptivity hierarchy theorems in property testing~\cite{CanonneG18}; on identifying biased coins through pairwise comparisons or in multi-armed bandit settings~\cite{AgarwalAAK17}; and on finding approximately maximum bipartite matchings through demand queries and OR-queries~\cite{Nisan21}. Other works have studied adaptivity in the massive parallel communication/computation (MPC) model~\cite{BeameKS13} and in various graph query models~\cite{AbasiB19,BeameHRRS20}.

A rich body of work on cost/adaptivity tradeoffs is found in {\em communication complexity}, where adaptivity manifests as rounds of interaction. An early work~\cite{NisanW93} gave exponential separations between $k$ and $k+1$ rounds for all $k$ and introduced a round elimination paradigm that remains ubiquitous to this day. This work also explains how an earlier result~\cite{KarchmerW90} connecting circuit and communication complexities can be used to relate bounded-round communication complexity for a specific problem to the size of bounded-depth, unbounded fan-in formulas. More work has been spurred by applications of bounded-round communication lower bounds in data structures, where they provide lower bounds in the cell-probe model~\cite{MiltersenNSW98,Sen03,ChakrabartiR04,PatrascuT07,LiuPY16}; and in streaming algorithms, where they translate naturally to tradeoffs between the number of passes made over the input stream and the working memory required~\cite{GuhaM09,EmekR14,GuruswamiOnak16,ChakrabartiCM16,ChakrabartiW16}. In much of this body of work, round elimination is performed using {\em information theoretic} arguments that naturally provide lower bounds against randomized algorithms.

In contrast, it is rare to see deterministic tradeoffs where corresponding randomized ones do not hold because randomization makes the problem ``too easy.'' This is exactly the situation with $\elemx$, as shown by this work in the context of the randomized upper bounds (\Cref{sec:upper-bounds}) via $\ell_0$-sampling~\cite{JowhariST11}. In light of the preceding discussion, our instantiations of round elimination must use techniques beyond Shannon-style information theory. They indeed do. Our techniques therefore have the potential for further use in separating determinism from randomization in this fine-grained (round aware) sense.

Our query complexity results on $\elemx$ suggest a tantalizing communication complexity analogue. Let $\urs_n$ denote\footnote{The notation, due to Nelson and Yu~\cite{NelsonY19}, is to be read as ``universal relation with a subset constraint.''} the communication complexity problem where Alice and Bob receive sets $X, Y \subseteq [n]$ respectively with the promise that $Y \subset X$, and their goal is to produce an element in $X \setm Y$.
Clearly, a $k$-round query protocol for $\elemx$ making $q$ queries, with each answer lying in a set of size $M$, provides a $k$-round communication protocol for $\urs$ using at most $q\log M$ bits. Therefore, our results here would be subsumed, in an asymptotic sense, if one could resolve the following conjecture positively.

\begin{conj}
  Every deterministic $k$-round communication protocol for $\urs$ costs $\tOmega(n^{1/k})$ bits, suppressing polylogarithmic factors.
\end{conj}

We find the above conjecture compelling because it would demonstrate a new phenomenon in communication complexity, where a problem is easy for one-round randomized {\em and} for interactive deterministic protocols, but exhibits a nontrivial tradeoff for bounded-round deterministic ones. 

In passing, we note that the $\urs$ problem was introduced in \cite{KapralovNPWWY17} where its randomized communication complexity was studied. The randomized lower bound was subsequently used by Nelson and Yu~\cite{NelsonY19} to prove the optimality of Ahn, Guha, and McGregor's graph sketching algorithm for graph connectivity~\cite{AhnGM12}. An outstanding open question about the latter problem (viewed as a communication problem where $n$ players, each holding a vertex neighborhood, talk to a coordinator who determines whether the graph is connected) is whether it admits a deterministic algorithm with sublinear communication. A better understanding of $\urs$ in the deterministic setting could be key to addressing this question.

There are also two problems similar to $\urs$ for which lower bounds have already been proven. The universal relation problem $\ur$ gives Alice and Bob unequal sets $X,Y \subseteq [n]$ and asks them to produce an element $i \in (X \setm Y) \cup (Y \setm X)$. This has deterministic communication complexity $\ge n +1$~\cite{TardosZ97}. The Karchmer-Wigderson game for $\textsc{parity}_n$ is the problem $\ur$ with the additional constraints that $|X|$ be even and $|Y|$ be odd; existing circuit complexity results~\cite{Hastad86,Rossman15} imply, as briefly explained in \Cref{sec:kw-lb}, that $k$-round deterministic communication protocols for this require $\Omega(k(n^{1/k}-1))$ bits of communication.

\section{Preliminaries} \label{sec:prelim}

Throughout the paper, we shall freely switch between the equivalent viewpoints of sets in $2^{[n]}$ and vectors in $\b^n$, using the notational convention that when an uppercase letter (e.g., $S,Z$) denotes a set, the corresponding lowercase boldface letter (e.g., $\bs,\bz$) denotes the characteristic vector of that set and vice versa.

\subsection{Various Definitions} \label{sec:defs}

The search problem $\elemx_n$ was already formally defined in \cref{eq:elemx-def-set}. We shall also work with special cases of this problem, where the cardinality of the input set is further restricted in some way. These are formalized as follows: we define
\begin{align}
  \elemxq_n &= \left\{ (Z,i) : Z \subseteq [n],\, i \in [n],\, \text{and } |Z| \not\equiv 0 \tpmod q \Rightarrow i \in Z \right\} \,; \label{eq:elemxq-def} \\
  \elemxqh_n &= \left\{ (Z,i) : Z \subseteq [n],\, i \in [n],\, \text{and } |Z| \equiv h \tpmod q \Rightarrow i \in Z \right\} \,; \label{eq:elemxqh-def} \\
  \elemxquarter_n &= \left\{ (Z,i) : Z \subseteq [n],\, i \in [n],\, \text{and } |Z| = n/4 \Rightarrow i \in Z \right\} \label{eq:elemxquarter-def} \,.
\end{align}

\begin{definition}[Protocol] \label{def:protocol}
Let $f \subseteq \b^n \times \cO$ be a search problem with input space $\b^n$ and output space $\cO$. A deterministic \emph{$k$-round $D$-linear query protocol} ($D$-LQP), $\Pi$, on this input space is a rooted tree of depth $k$ where each internal node $v$ is labeled with a matrix $A_v \in D^{d_v \times n}$; each leaf node with an output $o_\lambda \in \cO$; and the edges from a node $v$ to its children are labeled with the elements of $\cM_v := \{A_v \bz:\, \bz \in \b^n\}$ bijectively. The quantity $d_v$ of node $v$ is the \emph{cost} of the node, sometimes also denoted $\cost(v)$. Given an input $\bz \in \b^n$, the \emph{measurement at} internal node $v$ is $A_v \bz$. The transcript of $\Pi$ on $\bz$ ---denoted $\Pi(\bz)$---is the unique root-to-leaf path obtained by walking along the edges determined by these measurements; the \emph{$j$th measurement} is the label of the $j$th edge on this path; and the \emph{output} is the label $o_\ell$ of the leaf $\ell := \ell(\Pi(\bz))$ reached by this path. We say that $\Pi$ solves $f$ if $(\bz,o_\ell) \in f$ for every input $\bz$.

Since this paper is largely focused on deterministic complexity, henceforth we shall assume that all LQPs are deterministic unless stated otherwise.
\end{definition}

\begin{definition}[Cost]
The \emph{query cost} of a protocol $\Pi$ is:
\begin{align*}
  \cost(\Pi) &:= \max_{\bz \in \b^n} \cost(\Pi;\bz) \,,
  \qquad \text{where} \quad
  \cost(\Pi;\bz) := \sum_{v \text{ internal node on } \Pi(\bz)} d_v \,, 
  %\sum_{v \text{ internal node on } \Pi(\bz)} \cost(v) \,, \\
\end{align*}
which is, informally, the number of linear queries performed when $\Pi$ executes on $\bz$. While we do not focus on {\em bit} complexity in this paper, it is worth noting that to make an information-theoretically fair comparison between different domains, one should consider the number of bits returned in response to all the queries. This number may be larger than $\cost(\Pi)$, though only by an $O(\log n)$ factor for $D = \ZZ_{[-B,B]}$ with $B = \poly(n)$, and not at all for $D = \ZZ_2$.
\end{definition}

\begin{definition}[Complexity] \label{def:lq}
The {\em $D$-linear query complexity} and $k$-round $D$-linear query complexity of a search problem $f$ are defined, respectively, to be
\begin{align*} 
  \LQ_D(f) &= \min\{\cost(\Pi): \Pi \text{ is a $D$-LQP that solves } f\} \,; \\
  \LQ_D^k(f) &= \min\{\cost(\Pi): \Pi \text{ is a $k$-round $D$-LQP that solves } f\} \,.
\end{align*}
\end{definition}

\subsection{Useful Results from Combinatorics} \label{sec:comb}

In the course of this paper, we will use several important theorems from combinatorics. For results on $\ZZ_q$-LQPs (proved in \Cref{sec:zq-linear-measurements}), we use the following result of van Emde Boas and Kruyswijk \cite{EmdeBoasK69} on zero sumsets, slightly reworded to use modern notation.

\begin{theorem}[\cite{EmdeBoasK69}] \label{thm:emdeboas}
Let $G$ be a finite abelian group with exponent\footnote{The exponent of a group is the least common multiple of the orders of its elements.} $\exp(G)$ and order $|G|$. Let $s(G)$ be the minimal positive integer $t$ for which any sequence of $t$ elements from $G$ has a nonempty subsequence which sums to zero. Then $s(G) \le \exp(G) (1 + \ln\frac{|G|}{\exp(G)})$. \qed
\end{theorem}

A stronger result that $s(G) = 1 + r (q-1)$ applies when $G = \ZZ_{q}^r$ and $q$ is a prime power~\cite{Olson69}; it is conjectured that the prime-power constraint is unnecessary~\cite[conjecture~3.5]{GaoG06}.% conjecture is numbered 3.5

When working over $\ZZ$ (in \Cref{sec:int-linear-measurements}), we use the well-known notion of a \emph{sunflower} and the following recent result of Rao~\cite{Rao20}, which refines the noted result of Alweiss, Lovett, Wu, and Zhang~\cite{AlweissLWZ20} that improved the classic sunflower lemma of Erd\H{o}s and Rado~\cite{ErdosR60}. The note~\cite{BellCW20} further improves Rao's bound by replacing the $\log(pt)$ factor with $\log t$, but this will not affect our proof. Tao~\cite{Tao20} gives an alternative presentation of Rao's result which may be simpler to follow.

In a different part of our argument, we will need a well known theorem of Frankl and Wilson~\cite{FranklW81}.

\begin{theorem}[Rao] \label{thm:sunflower}
  There is a universal constant $c_1 > 1$ such that every family of more than $(c_1 p \log(pt))^t$ sets, each of cardinality $t$, must contain a \emph{$p$-sunflower}, defined as a family of $p$ distinct sets whose pairwise intersections are identical. \qed
\end{theorem}

\begin{theorem}[Frankl--Wilson] \label{thm:frankl-wilson}
Let $m(n,k,\overline{l})$ be the largest size of a collection $\mathcal{F}$ of subsets of $\binom{[n]}{k}$ for which no two elements $F,F' \in \mathcal{F}$ have intersection size $l$. Then, if $k - l$ is a prime power:
\begin{alignat*}{2}
  m(n,k,\overline{l}) & \le \binom{n}{k - l -1} \,, 
    &\qquad& \text{ if } k \ge 2 l + 1 \,; \\
  m(n,k,\overline{l}) & \le \binom{n}{l}\binom{2k - l - 1}{k} \Big/ \binom{2k - l -1}{l}\,, 
    && \text{ if } k \le 2 l + 1 \,. \rlap{~~\qquad\qquad\qquad\qquad\qed}
    %% Is there a cleaner way to get a flushright qed symbol here?
\end{alignat*}
\end{theorem}

\subsection{Our Round Elimination Framework} \label{sec:framework}

We now describe a framework for our round elimination arguments. For this
section, we shall work over a general ring (with unity), $R$, and ``LQP'' will mean a $D$-LQP where $D \subseteq R$. Fix this ring $R$.
% R must have a unit, so lifting matrix exists; we do NOT actually use the facts that {0,1} are in D

\begin{definition}[Homomorphism and shadowing]
A {\em protocol homomorphism} is a map $\varphi$ from a protocol $\Upsilon$ to a protocol $\Pi$ such that (i)~for any two nodes $u,v$ in $\varphi$, the node $\varphi(u)$ is a child of $\varphi(v)$ iff $u$ is a child of $v$, and (ii)~$\varphi$ maps leaves of $\Upsilon$ to leaves of $\Pi$.  We say that $\varphi$ is {\em cost-preserving} for each internal node $v$ of $\Upsilon$, $\cost(v) = \cost(\varphi(v))$.  We say that $\Upsilon$ {\em shadows $\Pi$ through $\varphi$} if $\varphi$ is injective, cost-preserving, and maps the root of $\Upsilon$ to a child of the root of $\Pi$. Notice that when this is the case, $\Upsilon$ is one round shorter than $\Pi$.
\end{definition}

Suppose we have an LQP $\Pi$ that operates on inputs in $\b^n$ and produces outputs in $[n]$. Further, suppose $S_1, \ldots, S_m \subseteq [n]$ is a collection of pairwise disjoint nonempty sets. We then define a certain LQP $\Pi^{(S_1,\ldots,S_m)}$ operating on inputs in $\b^m$ and producing outputs in $[m]$. To aid intuition, we describe the construction procedurally in \Cref{alg:lift-protocol}.

\begin{algorithm*}[!ht]
  \caption{~Outline of protocol $\Pi^{(S_1,\ldots,S_m)}$}
  \label{alg:lift-protocol}
  \begin{algorithmic}[1] % TODO: use section numbering.
    \State Lift our input $W \subseteq [m]$ to $Z := \bigcup_{i \in W} S_i \subseteq [n]$ (this step is only conceptual). \label{line:pi-lift}
    \State Mimic $\Pi$ by simulating the queries it would have made to its input $Z$. Emulate each such query by making the corresponding query to our own input $W$. This is indeed possible using linear queries to $W$. \label{line:pi-emu}
    \State Suppose $\Pi$ wants to output $h$. If $h \in S_i$, then output that index $i$ (which must be unique); otherwise, output an arbitrary index. \label{line:pi-output}
  \end{algorithmic}
\end{algorithm*}

To define $\Pi' := \Pi^{(S_1,\ldots,S_m)}$ formally, we first define the \emph{lifting matrix}
\begin{align} \label{eq:lift-mtx}
  L = [\bs_1 ~ \bs_2 ~ \cdots ~ \bs_m] \in R^{n \times m} \,,
\end{align}
whose entries lie in $\b$ and which maps the input space of $\Pi'$ to the input space of $\Pi$ according to \cref{line:pi-lift}, thanks to the pairwise disjointness of the sets $S_i$. At a given node $v$ of $\Pi$, labeled with $A_v \in \ZZ_q^{d_v \times n}$, the simulation in \cref{line:pi-emu} would retrieve the measurement $A_v \bz = A_v L \bw$. The protocol $\Pi'$ can get the same result by making the query $A_v L \in \ZZ_q^{d_v \times m}$.

Thus, the protocol tree for $\Pi'$ is formed as follows. Prepare a copy of $\Pi$ and let $\varphi \colon \Pi' \to \Pi$ be the natural bijection between their nodes. Label each internal node $v$ of $\Pi'$ with $A_v := A_{\varphi(v)} L$. Copy over all edge labels from $\Pi$ to $\Pi'$. For each leaf $\ell$ of $\Pi'$, if $o_{\varphi(\ell)} \in S_i$, then assign label $o_\ell := i$. If no such $i$ exists, assign $o_\ell := 1$ (say). This labeling is well defined because of the pairwise disjointness of the sets $S_i$.

\medskip
In the sequel, to perform round elimination, we shall use the construction of $\Pi'$ in a special way that we record in the lemma below. We also record a definition that will be relevant when invoking the lemma.

\begin{lemma} \label{lem:round-el-gen}
Suppose that $\Pi$ correctly solves $\elemx_n$ on inputs in $\cZ \subseteq \b^n$. 
Let $S_1,\ldots,S_m \subseteq [n]$ be pairwise disjoint and let $L$ be defined by \cref{eq:lift-mtx}.
Let $\rho$ be the root node of $\Pi$ and, for $\br \in R^{d_\rho}$, let $\cW_\br := \{\bw \in \b^m:\, L\bw \in \cZ$ and $A_\rho L\bw = \br\}$.
%$\cZ_\br := \{\bz \in \cZ:\, A_\rho \bz = \br\}$. 
Then, there is a protocol $\Upsilon$ that shadows $\Pi$ and correctly solves $\elemx_m$ on each input in $\cW_\br$.
\end{lemma}
\begin{proof}
  Using the above setup and terminology, construct $\Pi' := \Pi^{(S_1,\ldots,S_m)}$ as in \Cref{alg:lift-protocol}. The given conditions imply that on all inputs in $\cW_\br$, the first measurement of $\Pi'$ is always $\br$ and thus leads an execution of $\Pi'$ to a particular child, $u$, of its root node. Thus, we can shrink $\Pi'$ to the subprotocol $\Upsilon$ rooted at $u$. Notice that the bijection $\varphi$ is a cost-preserving protocol homomorphism and so $\Upsilon$ shadows $\Pi$ through $\varphi|_{\Upsilon}$.
  
  By construction, $\Upsilon$ on input $\bw \in \cW_\br$ simulates $\Pi$ on $\bz := L\bw = \sum_{i \in W} \bs_i$, an input on which $\Pi$ correctly solves $\elemx_n$. Therefore, if $\Pi$ outputs $h$, then $h \in Z = \bigcup_{i \in W} S_i$. By the disjointness guarantee, there exists a unique $i \in W$ for which $h \in S_i$. As $\Upsilon$ reports precisely this $i$, it correctly solves $\elemx_m$ on $\bw$.
\end{proof}

% Defining A-uniform only now, because doing it earlier might mislead readers
% into thinking that we needed uniformity for the above lemma.
\begin{definition}[Uniform family]
Fix a matrix $A \in R^{d \times n}$. An {\em $A$-uniform family of size $m$} is a collection of $m$ pairwise disjoint sets $S_1, \ldots, S_m \subseteq [n]$ such that $A\bs_1 = \cdots = A\bs_m = \br$, for some vector $\br \in R^d$.
\end{definition}

\section{Linear Queries Modulo 2} \label{sec:z2-linear-measurements}

We begin our study of the {\sc element-extraction} problem by considering $\ZZ_2$-linear queries. As noted in \Cref{sec:results}, we shall later generalize the results to $\ZZ_q$, but we feel it is worth seeing our framework in action in the especially clean setting of $\ZZ_2$. We begin by showing that the additional promise of odd cardinality on the input set $Z$ is crucial, or else there is no interesting rounds-vs-queries tradeoff to be had.

\begin{proposition}[Restatement of \Cref{thm:z2-hard}]
  $\LQ_{\ZZ_2}(\elemx_n) = n-1$.
\end{proposition}
\begin{proof}
The upper bound is achieved by the trivial $1$-round LQP (i.e., a sketch) that queries all but one of the individual bits of the input.

Now assume to the contrary that there is a $\ZZ_2$-LQP $\Pi$ with $\cost(\Pi) = d \le n-2$ that solves $\elemx_n$. Let $A \in \ZZ_2^{d \times n}$ be the matrix whose rows represent all queries along the path $\Pi(\bzero)$. Then $\dim\,\ker A \ge n-d \ge 2$, whence there exist distinct nonzero vectors $\by,\bz \in \ZZ_2^n$ such that $A\by = A\bz = \bzero$. Setting $\bx = \by + \bz$, we also have $A\bx = \bzero$. Thus, the three nonzero inputs $\bx, \by, \bz$ lead to the same leaf, namely $\ell(\Pi(\bzero))$, and produce the same output $i$, say. By the correctness of $\Pi$, we have $x_i = y_i = z_i = 1$, which contradicts $\bx = \by + \bz$. 
\end{proof}

Accordingly, for the rest of this section, we focus on the problem $\elemxodd_n$, as defined in \cref{eq:elemxq-def}. We shall prove \Cref{thm:z2-k-round} using a round elimination technique. As discussed in \Cref{sec:prelim}, this round elimination will be enabled by identifying a certain $A$-uniform family. The next lemma, establishing a useful fact about matrices over $\ZZ_2$, will provide us this family.

\begin{lemma} \label{clm:z2-partition}
  Every matrix $A \in \ZZ_2^{d \times n}$, admits an $A$-uniform family $S_1, \ldots, S_m$ of size $m \ge \ceil{n/(d+1)}$ such that each cardinality $|S_i|$ is odd.
\end{lemma}

\begin{proof}\label{proof:z2-matrix-lemma}
  Let $\bb_1, \ldots, \bb_n$ be the (nonzero) column vectors of the matrix
  \[
    B := \left[\begin{array}{c} A\\ \bone^\txp \end{array}\right] \in \ZZ_2^{(d+1) \times n}
  \]
  formed by appending the all-ones row to $A$. For each $Q \subseteq [n]$, let $B_Q$ be the collection of column vectors $\{\bb_i:\, i \in Q\}$ and let $\gen{B_Q}$ be the linear subspace of $\ZZ_2^{d+1}$ spanned by the vectors in $B_Q$.
  
  Partition $[n]$ into nonempty disjoint sets $T_1, \ldots, T_m$ iteratively, as follows. For each $i$, let $T_i$ be a {\em maximal} subset of $[n] \setm \bigcup_{j=1}^{i-1} T_j$ such that the vectors in $B_{T_i}$ are linearly independent. Since these vectors live in $\ZZ_2^{d+1}$, it follows that $|T_i| \le d+1$. We stop when $\bigcup_{j=1}^m T_m = [n]$, implying $m \ge \ceil{n/(d+1)}$.
  
  We claim that, for each $i \in \{2,\ldots,m\}$, we have $\gen{B_{T_{i-1}}} \supseteq \gen{B_{T_i}}$. Indeed, if there exists an element $\bx \in \gen{B_{T_i}}\setm \gen{B_{T_{i-1}}}$, then there is a set $Q \subseteq T_i$ for which $\bx = \sum_{h \in Q} \bb_j$. Since $\gen{B_{T_{i-1}}}$ is closed under linear combinations and does not contain $\bx$, there exists $h \in Q$ with $\bb_h \notin \gen{B_{T_{i-1}}}$. By construction, $h \notin \bigcup_{j=1}^{i-2} T_j$, so $h$ was not included in $T_{i-1}$ despite being available. This contradicts the maximality of $T_{i-1}$.
  
  Let $k$ be an index in $T_m$. Then $\bb_k \in \gen{B_{T_m}} \subseteq \gen{B_{T_{m-1}}} \subseteq \cdots \subseteq  \gen{B_{T_1}}$, so there must exist subsets $S_1, \ldots, S_m$ of $T_1, \ldots, T_m$ for which $B \bs_i = \bb_k$. The sets $\{S_i\}_{i=1}^{m}$ are pairwise disjoint because the sets $\{T_i\}_{i=1}^{m}$ are. Let $\br$ be the first $d$ coordinates of $\bb_k$; then for all $i \in [m]$, $A \bs_i = \br$. Therefore, $\{S_i\}_{i=1}^{m}$ is $A$-uniform. Finally, since the last coordinate of $\bb_k$ is $1$ and the last row of $B$ is $\bone^\txp$, for each $i \in [m]$, $\bone^\txp \bs_i = 1$, so $|S_i|$ is odd.
\end{proof}

\begin{lemma}[Round elimination lemma] \label{lem:z2-round-elim}
  Let $\Pi$ be a deterministic $k$-round $\ZZ_2$-LQP for $\elemxodd_n$, where $k \ge 1$. Then there exists a deterministic $(k-1)$-round $\ZZ_2$-LQP $\Upsilon$ for $\elemxodd_m$, such that
  \begin{thmparts}
    \item $\Upsilon$ shadows $\Pi$ through a (cost-preserving, injective) protocol homomorphism $\varphi_\Upsilon \colon \Upsilon \to \Pi$; \label{clm:z2-costs}
    \item $m \ge \ceil{n/(d + 1)}$, where $d$ is the cost of the root of $\Pi$. \label{eq:z2-shrinkage}
  \end{thmparts}
\end{lemma}
\begin{proof}
  Let $A \in \ZZ_2^{d \times n}$ be the label of the root of $\Pi$.
  Let $S_1, \ldots, S_m$ be an $A$-uniform family of size $m \ge \ceil{n/(d+1)}$ with each $|S_i|$ odd, as guaranteed by \Cref{clm:z2-partition}. Let the lifting matrix $L$ be as given by \cref{eq:lift-mtx} and let $\br = A\bs_1$.
  We know that $\Pi$ correctly solves $\elemxodd_n$ on inputs in $\cZ := \{Z \subseteq [n]:\, |Z|$ odd$\}$. 
  Invoking \Cref{lem:round-el-gen}, we obtain a $(k-1)$-round $\ZZ_2$-LQP $\Upsilon$ that shadows $\Pi$ as required.
  
  It remains to show that $\Upsilon$ solves $\elemxodd_m$. The guarantee of \Cref{lem:round-el-gen} is that $\Upsilon$ correctly solves $\elemx_m$ on the input set $\cW_\br$ defined there. Thus, it suffices to show that if an input $W \subseteq [m]$ satisfies the promise of $\elemxodd_m$---i.e., $|W|$ is odd---then $W \in \cW_\br$. We reason as follows:
  \begin{alignat*}{2}
    |W| \text{ odd} 
    &\implies |L\bw| = \left|\sum_{i \in W} \bs_i\right| \equiv 1 \tpmod 2 &\qquad& \lhd \text{ each $|S_i|$ is odd} \\
    &\implies L\bw \in \cZ \,; && \lhd \text{ definition of $\cZ$} \\
  \intertext{and}
    |W| \text{ odd}
    &\implies A L\bw = A \sum_{i \in W} \bs_i = |W|\cdot A\bs_1 = |W| \cdot \br = \br \,. && \lhd \text{ definition of $A$-uniformity}
  \end{alignat*}
  This completes the proof, by definition of $\cW_\br$.
\end{proof}

The next step of the proof is to repeatedly invoke the above round elimination lemma and carefully control parameters. To perform a sharp analysis, we introduce the following concept.

\begin{definition} \label{def:division-seq}
A {\em division sequence} for $n$ is a finite sequence of positive integers $d_1 \ldots d_{j}$ for which
\begin{align}\label{eq:division-seq-precondition} 
  \left\lceil \cdots \left\lceil \left\lceil 
    n \cdot \frac{1}{d_1 + 1} \right\rceil \frac{1}{d_1 + 1} \right\rceil \cdots \frac{1}{d_{j} + 1}
    \right\rceil = 1 \,.
\end{align}
\end{definition}
\begin{lemma}\label{lem:division-seq}
Let $d_1, \ldots, d_{j}$ be a division sequence for $n$ minimizing $\sum_{h=1}^{j} d_h$. Then
\begin{align*}
  j n^{1/j} - j \le \sum_{h=1}^{j} d_h \le j \lceil n^{1/j} \rceil - j \,.
\end{align*}
\end{lemma}
\begin{proof}
For the upper bound, let $d_1 = \ldots = d_{j} = \lceil n^{1/j} \rceil - 1$. For the lower bound, remove the ceiling operations in \cref{eq:division-seq-precondition} to get
\begin{align*}
    \frac{n}{\prod_{h=1}^{j} (d_h + 1)} \le 1 \,,
    & \qquad \text{which implies} \quad n^{1/j} \le \smash{\left(\prod_{h=1}^{j} (d_h + 1)\right)^{1/j}} \,.
\end{align*}
By the AM-GM inequality,
\[
    \sum_{h=1}^{j} d_h
    = j \left(\frac{1}{j}\sum_{h=1}^{j} (d_h + 1) - 1\right)
    \ge \smash{j \left( \left(\prod_{h=1}^{j} (d_h + 1)\right)^{1/j} - 1 \right)}
    \ge j(n^{1/j} - 1) \,. \qedhere
\]
\end{proof}

This brings us to the main result of this section: a rounds-vs-queries tradeoff.

\begin{theorem}[Restatement of \Cref{thm:z2-k-round}]
  $\LQ_{\ZZ_2}^k(\elemxodd_n) \ge k(n^{1/k} - 1)$.
\end{theorem}
\begin{proof}
  Suppose that $\Pi$ is a deterministic $k$-round $\ZZ_2$-LQP for $\elemxodd_n$. Repeatedly applying \Cref{lem:z2-round-elim}, we obtain a sequence of protocols $\Pi=\Pi_1, \Pi_2, \ldots, \Pi_{j+1}$, which solve $\elemxodd$ on progressively smaller input sizes, until $\Pi_{j+1}$ is a degenerate depth-0 protocol (in which no queries occur).

  Let $d_i$ be the cost of the root $\rho_i$ of $\Pi_i$, for $1 \le i \le j$. As \Cref{clm:z2-costs} gives protocol homomorphisms $\varphi_{\Pi_{i+1}} : \Pi_{i+1} \to \Pi_i$, we find the the roots of each $\Pi_i$ correspond to nodes $u_i = (\varphi_{\Pi_2} \circ \cdots \circ \varphi_{\Pi_i})(\rho_i)$ in $\Pi$. In fact, the vertices $u_1,u_1,\ldots,u_{j+1}$ form a path from the root $\rho = u_1$ of $\Pi$ to the leaf $u_{j+1}$. The inputs of $\Pi_{j+1}$ lift to inputs of $\Pi$ which reach $u_{j+1}$. Lower bounding the query cost of $\Pi$ using this branch gives
  \begin{align} \label{eq:z2-lincost-sum}
    \cost(\Pi) \ge \sum_{i=1}^{j} \cost(u_i) = \sum_{i=1}^{j} d_i  \,.
  \end{align}

  Using \cref{eq:z2-shrinkage} repeatedly, $\Pi_{j+1}$ must solve $\elemxodd_m$, for some integer \begin{align*}
    m \ge \left\lceil \cdots \left\lceil \left\lceil n \cdot \frac{1}{d_1 + 1} \right\rceil \frac{1}{d_2 + 1} \right\rceil \cdots \frac{1}{d_j + 1}  \right\rceil \,.
  \end{align*}
  However, as $\Pi_{j+1}$ solves $\elemxodd_m$ without performing any queries, there must be a fixed index which is a valid output for all inputs $Z \in 2^{[m]}$ of odd size. This is only possible when $m = 1$; for any larger $m$, the inputs $Z = \{1\}$ and $Z' = \{2\}$ must produce different outputs.
  
  Therefore, the integers $d_1, \ldots, d_j$ form a division sequence for $n$. Applying \Cref{lem:division-seq} to \cref{eq:z2-lincost-sum},
  \begin{align*}
    \cost(\Pi) \ge \sum_{i=1}^{j} d_i \ge j n^{1/j} - j \ge k(n^{1/k} - 1) \,,
  \end{align*}
  where the last inequality follows from the fact that $\frac{d}{dz} \left[z (n^{1/z} - 1)\right] \le 0$ for all $z\ge0$.
\end{proof}

\section{Linear Queries Modulo \textit{\textbf{q}}} \label{sec:zq-linear-measurements}

First, we use \Cref{thm:emdeboas} to show that $\elemx_n$ is hard for $\ZZ_q$-LQPs.

\begin{proposition}[Restatement of \Cref{thm:zq-hard}]
  For every $q \ge 3$, we have $\LQ_{\ZZ_q}(\elemx_n) \ge n/(2q \ln q)$.
\end{proposition}
\begin{proof}
  This is proven with the same strategy as for \Cref{thm:z2-hard}. Assume for sake of contradiction that $\cost(\Pi) \le \frac{n}{2 q \ln q}$. Let $\nu$ be the leaf $\ell(\Pi(\bzero))$. Let $A \in \ZZ_q^{d \times n}$ be the matrix containing all queries along the path from the root of $\Pi$ to $\nu$.
  
  By \Cref{thm:emdeboas}, since the group $\ZZ_q^d$ has order $q^d$, and exponent $q$, any sequence of $D \le q \left(1 + \ln(\frac{q^d}{q}) \right)$ elements in $\ZZ_q^d$ has a nontrivial subsequence summing to $\bzero$. As $q \ge 3$, $d q \ln q \ge D$. Thus, since $n \ge 2 d q \ln q$, picking disjoint subsets $I$ and $J$ of sizes $d q \ln q$ each, and applying the theorem implies there exist disjoint nonempty subsets $Z_1$ and $Z_2$ of $[n]$ for which the corresponding columns of $A$ sum to $\bzero$. In other words, $\Pi$ reaches the same leaf given $\bz_1$ and $\bz_2$, but the leaf cannot be assigned an output consistent with both.
\end{proof}

A similar strategy proves a lemma analogous to \Cref{clm:z2-partition}:

\begin{lemma} \label{clm:zq-partition}
  Every matrix $A \in \ZZ_q^{d \times n}$, admits an $A$-uniform family $S_1, \ldots, S_m$ where
  \begin{thmparts}
    \item $\displaystyle m \ge \frac{n}{(d+1) q \ln q} - 1$, and \label{clm:zq-m-size}
    \item each cardinality $|S_i| \equiv -1 \pmod q$. \label{clm:zq-partitions-parity}
  \end{thmparts}
\end{lemma}
\begin{proof}
To be able to enforce constraints on the values $|S_i|$, we define $B := \left[\bone \mid A^\txp\right]^\txp \in \ZZ_q^{(d+1) \times n}$, and let $\bb_1,\ldots,\bb_n$ be its column vectors. We partition the columns of the matrix $B$ into disjoint subsets $D_1,\ldots,D_k$ of $[n]$ by the following iterative procedure. In the procedure, let $P$ be the set of indices of $[n]$ not yet chosen. Each set $D_i$ starts out as $\varnothing$; then beginning with $i=1$, each set $D_i$ is expanded by picking an index $j$ from $P$ for which no subset $H \subseteq (D_i \cup \{j\})$ has the property that $\sum_{h \in H} \bb_h = \bzero$; adding $j$ to $D_i$ and removing $j$ from $P$; until no more such indices can be found. When $D_i$ is done, start filling $D_{i+1}$, etc.

When $q=2$, each $D_i$ corresponds to a basis of a subspace of $\ZZ_2^{d+1}$, so $|D_i| \le d+1 < (d+1) 2 \ln 2$. For $q \ge 3$, we apply \Cref{thm:emdeboas}, using the fact that the group $\ZZ_q^{d+1}$ has order $q^{d+1}$ and exponent $q$. The maximum possible size of each set $D_i$ is then $\le q \left(1 + \ln(\frac{q^{d+1}}{q}) \right) - 1$. The upper bound $(d+1) q \ln q$ also holds here. Consequently, the number $k$ of sets formed is $\ge \frac{n}{(d+1) q \ln q}$. Pick some $t \in D_k$; for any $i < k$, since $t$ was not picked when $D_i$ was constructed, it must be the case that there is a subset $S_i \subseteq D_i$ for which $\sum_{h \in S_i} \bb_h + \bb_t = \bzero$. This implies $B \bs_i = \sum_{h \in S_i} \bb_h = -\bb_t$. Since the first row of $B$ is $\bone$, we have $|S_i| \equiv \sum_{h \in S_i} 1 \equiv -1 \pmod q$, so all the sets $S_i$ have size $-1 \pmod q$. Let $\br$ be the last $d$ entries of $-\bb_t$; then for all $i$, $B \bs_i = \br$. There are $m = k - 1 \ge \frac{n}{(d+1) q \ln q} - 1$ sets in total.
\end{proof}

Compared to $\elemxodd$, there is a slight weakening of the main round elimination lemma, which is a direct consequence of the weakened \Cref{clm:zq-partition}. Instead of directly lower bounding the cost of $\elemxq_n$, we prove separate lower bounds for each $\elemxqh_n$, for all $h \in \{1,\ldots,q-1\}$, and take their maximum. The search problem $\elemxqh_n$ is $\elemx_n$ with the additional promise that the input set $Z$ has size $\equiv h \pmod q$.

\begin{lemma}[Round elimination lemma] \label{lem:zq-round-elim}
  Let $\Pi$ be a $k$-round $\ZZ_q$-LQP for $\elemxqh_n$, where $k \ge 1$ and $h \in \{1,\ldots,q-1\}$. Then there exists a $(k-1)$-round $\ZZ_q$-LQP $\Upsilon$ for $\elemx_m^{(q, -h)}$, such that
  \begin{thmparts}
    \item $\Upsilon$ shadows $\Pi$ through a protocol homomorphism $\varphi_\Upsilon \colon \Upsilon \to \Pi$; \label{clm:zq-costs}
    \item $\displaystyle  m \ge \frac{n}{(d+1) q \ln q} - 1$, where $d$ is the cost of the root of $\Pi$. \label{eq:zq-m-size-lem}
  \end{thmparts}
\end{lemma}

\begin{proof}
  Let $A \in \ZZ_q^{d \times n}$ be the label of the root of $\Pi$. \Cref{clm:zq-partition} guarantees that there exists an $A$-uniform family of size $m$, where $m$ satisfies \cref{eq:zq-m-size-lem}, and $A \bs_1 = \ldots = A \bs_m  = \bx$, and $|S_1| \equiv \ldots \equiv |S_m|  \equiv -1 \pmod q$. Let $L$ be the lifting matrix from \cref{eq:lift-mtx}, and $\br = -h \bx$. Applying \Cref{lem:round-el-gen} to $\Pi$, $L$ and $\br$, we obtain a $(k-1)$-round $\ZZ_q$-LQP $\Upsilon$ that shadows $\Pi$, and solves $\elemx_m$ on all inputs $W \subseteq [m]$ for which $A L \bw = \br$ and $|L \bw| \equiv h \pmod q$. If $W$ fulfills the promise of $\elemx^{(q,-h)}_n$, that $|W| \equiv -h \pmod q$, then:
  \begin{align*}
      |L \bw| &= \left|\bigcup_{i\in W} S_i\right| =\sum_{i \in W} |S_i| = |W| \cdot (-1) = h \pmod q \,, \\
      A L \bw &= \sum_{i\in W} A \bs_i = |W| \bx = -h \bx = \br \,,
  \end{align*}
  which proves that $\Upsilon$ is correct on $W$.
\end{proof}

This brings us to the main result of this section, which essentially generalizes the modulo-$2$ result from the previous section.

\begin{theorem}[Restatement of \Cref{thm:zq-k-round}]
  For each $q \ge 2$, we have
  \[
    \LQ_{\ZZ_q}^k(\elemxq_n) \ge \frac{1}{3.67 q^{1 + 1/k} \ln^2 q} k(n^{1/k} - 1)\,.
  \]
\end{theorem}
\begin{proof}
  Suppose that $\Pi$ is a deterministic $\ZZ_q$-LQP for $\elemxqh_n$.  Repeatedly applying \Cref{lem:zq-round-elim}, we construct a sequence of protocols $\Pi = \Pi_1, \Pi_2, \ldots, \Pi_{j+1}$, which respectively solve $\elemxqh$, $\elemx_n^{(q, -h)}$, $\elemxqh, \ldots\,$ on progressively smaller input sizes, until $\Pi_{j+1}$ is a degenerate depth-0 protocol (in which no queries occur), for $\elemx_n^{(q, (-1)^j h)}$. As in \Cref{sec:z2-linear-measurements}, the roots $\rho_i$ of the protocols $\Pi_i$, $1 \le i \le j$, which have cost $d_i$, correspond to a branch of $\Pi$ formed by corresponding nodes $u_i$ and ending at a leaf corresponding to the root of $\Pi_{j+1}$. Then
  \begin{align} \label{eq:zq-lincost-sum}
    \cost(\Pi) \ge \sum_{i=1}^{j} \cost(u_i) = \sum_{i=1}^{j} d_i  \,.
  \end{align}
  Let $\delta_i := (d_i + 1) q \ln q$.  By \Cref{lem:zq-round-elim}, $\Pi_i$ solves $\elemx^{(q,(-1)^{i-1} h)}_{m_i}$, where $m_1 = n$, and:
  \begin{align}
    m_{i+1} & \ge \frac{m_i }{(d_i+1) q\ln q} - 1 = \frac{m_i - \delta_i}{\delta_i} \,. \label{eq:rel-step}
  \end{align}
  As $\Pi_{j+1}$ solves $\elemx^{(q,(-1)^j h)}_{m_i}$ without any queries, the problem must be trivial, necessitating $m_{j+1} \le q$. Combining \cref{eq:rel-step} for $i$ between $1$ and $j$ and rearranging:
  \begin{align*}
    q \ge \frac{n - \sum_{i=1}^{j}{\prod_{\ell=1}^{i}{\delta_\ell}}}{ \prod_{\ell=1}^{j}{\delta_\ell} } \quad\implies\quad n &\le q \prod_{\ell=1}^{j}{\delta_\ell} + \sum_{i=1}^{j}{\prod_{\ell=1}^{i}{\delta_\ell}} \le (q+j) \prod_{\ell=1}^{j}{\delta_\ell} \,.
  \end{align*}
  Further rearrangement lets us use AM-GM and an inequality derived from $(q+j) q^j \le (q+1)^j q$:
  \begin{align}\label{eq:zq-lincost-amgm}
    \left(\frac{1}{j} \sum_{i=1}^{j}{\delta_i}\right) \ge \left(\prod_{i=1}^{j}{\delta_i}\right)^{1/j} \ge \left(\frac{n}{q+j}\right)^{1/j} \ge \frac{q}{q+1} \left(\frac{n}{q}\right)^{1/j} \,.
  \end{align}

  We can now lower bound the query cost of $\Pi$:
  \begin{alignat}{2}
    \cost(\Pi) &\ge \sum_{i=1}^{j} (d_i+1) - j = \frac{1}{q \ln q} \sum_{i=1}^{j} \delta_i - j &\qquad\qquad & \lhd \text{ by \cref{eq:zq-lincost-sum}} \nonumber\\
    &\ge j \left( \frac{1}{(q+1) \ln q} \left(\frac{n}{q}\right)^{1/j} -1 \right) && \lhd \text{ by \cref{eq:zq-lincost-amgm}} \nonumber\\
    &\ge k \left( \frac{1}{(q+1) \ln q} \left(\frac{n}{q}\right)^{1/k} -1 \right) \,. && \lhd \text{ since $\frac{d}{ds} [s (r^{1/s} - 1)] \le 0$}
  \label{eq:lb-for-zq-z-eq-h}
  \end{alignat}

  This lower bound becomes negative for sufficiently large $k$. To obtain a bound that remains positive for all $k$, we combine it with an unconditional lower bound. First, we note that \cref{eq:lb-for-zq-z-eq-h} also applies to protocols solving $\elemxq_n$, since $\elemxqh_n$ was an easier case. For $\elemxq_n$, the set of possible transcripts of any protocol $\Psi$ forms a $q$-ary prefix code of maximum length $d$. If $q^d < n$, then by the pigeonhole principle $\Psi$ must treat identically some pair of $\{1\},\{2\},\ldots,\{n\}$, which is a contradiction; thus $\cost(\Pi) \ge \ln n / \ln q$. Combining this lower bound with \cref{eq:lb-for-zq-z-eq-h} and applying \Cref{lem:knwk-lb}, we obtain
  \[
      \cost(\Pi) 
      \ge \max\left\{\frac{\ln n}{\ln q},\, k \left( \frac{1}{q^{1/k} (q+1) \ln q} n^{1/k} - 1 \right) \right\}
      \ge \frac{k (n^{1/k} - 1)}{q^{1/k} (q+1) (\ln q + 1) \ln q} \ge \frac{k (n^{1/k} - 1)}{3.67 q^{1 + 1/k} \ln^2 q} \,. \qedhere
  \]
\end{proof}

\section{Linear Queries Over the Integers} \label{sec:int-linear-measurements}

For $\ZZ$-LQPs, our main result is a 2-round lower bound for $\elemx_n$. We require a careful accounting of the query cost of a protocol, to adjust for the fact that the (bit) size of the query results
depends on the maximum entry value in a given query matrix. This motivates the following definition and observation.

\begin{definition}
  A $\ZZ$-LQP is said to be {\em $M$-bounded} if each linear measurement can take at most $M$ distinct values. In particular, if the inputs to a $\ZZ_{[-B,B]}$-LQP $\Pi$ lie in $\b^n$, then $\Pi$ is $(B n + 1)$-bounded.
\end{definition}

Recall the problem $\elemxquarter_n$ defined in \cref{eq:elemxquarter-def}. For $n$ divisible by $4$, this is simply $\elemx_n$ under the additional promise that $|Z| = n/4$. We first prove a $1$-round lower bound for this problem, under a slight additional assumption on $n$.

\begin{lemma} \label{lem:int-quarter-oneround}
Let $n = 4 r$ where $r$ is a prime power. If $\Pi$ is an $M$-bounded one-round protocol for $\elemxquarter_n$,
\begin{align*}
  \cost(\Pi) \ge 0.14 \frac{n}{\log M} \,.
\end{align*}
\end{lemma}

\begin{proof}
Let $d = \cost(\Pi)$ and let $A \in \ZZ^{d \times n}$ be the query performed by $\Pi$. We first consider what $\Pi$ does on inputs of cardinality $n/2$, even though such inputs lie outside the promise region of $\elemxquarter_n$. Soon, we shall see how this helps. 

Since $\Pi$ is $M$-bounded, the mapping $\bz \mapsto A \bz$ from domain $\binom{[n]}{n/2}$ to $\ZZ^d$ has no more than $M^d$ possible output values. By the pigeonhole principle, there exists a vector $\bw \in \ZZ^d$ for which 
\begin{align*}
  \cF_\bw := \left\{ \bz \in \binom{[n]}{n/2}:\, A \bz = \bw \right\} \qquad\text{ has }\qquad |\cF_\bw| \ge \binom{n}{n/2} M^{-d} \,.
\end{align*}
If there exist two distinct vectors $\bx,\by \in \cF_\bw$ such that $|\bx \cap \by| = n/4$, then we can construct two disjoint vectors which $\Pi$ can not distinguish, and thus cannot give a correct answer to $\elemxquarter_m$ in both cases. Specifically, $|\bx \setm \by| = |\by \setm \bx| = n/4$ and
\begin{align*}
  A (\bx \setm \by) = A \bx - A (\bx \cap \by) = \bw - A (\bx \cap \by) = A \by - A (\bx \cap \by) = A (\by \setm \bx) \,.
\end{align*}
By \Cref{thm:frankl-wilson}, if there does not exist such a pair $\bx,\by$, then we have an upper bound on $|\cF|$, and can derive
\begin{align*}
   \binom{n}{n/2} M^{-d} & \le |\cF| \le \binom{n}{n/4 - 1} \,.
\end{align*}
Therefore, $\binom{n}{n/2} M^{-d} \le \binom{n}{n/4}$ and we obtain
\[
   d \ge \frac{\log{\binom{n}{n/2}}- \log{\binom{n}{n/4}}}{\log M}
     \ge \frac{\log{\binom{4}{2}} - \log \binom{4}{1}}{4} \frac{n}{\log M} \ge 0.14 \frac{n}{\log M} \,. \qedhere
\]
\end{proof}

It should be noted that without the promise that $|Z| = n/4$, a one-round lower bound would follow very easily. By a standard ``decoding'' argument, a one-round protocol for $\elemx$ can be used to recover the entire unknown input $\bz$. For completeness, we give the easy proof below. The reason we needed the much more complicated argument in \Cref{lem:int-quarter-oneround} above is that the promise in $\elemxquarter_n$ prevents us from performing such a decoding.

\begin{proposition}[Essentially a restatement of \Cref{prop:int-one-round}]
If $\Pi$ is an $M$-bounded one-round protocol for $\elemx_n$, then $\cost(\Pi) \ge n/\log M - 1$.
\end{proposition}
\begin{proof}
Modify $\Pi$ to add the query $\bone^\txp$, which reports $|Z|$; this increases $\cost(\Pi)$ by one. Let $A \in \ZZ^{(d+1)\times n}$ be the modified query matrix. Since $\Pi$ is correct, $A \bz$ determines an index $i_1 \in Z$. Let $\be_{i_1}$ be the indicator vector for $i_1$; since we know $\be_{i_1}$, we can compute $A (\bz - \be_{i_1})$ without making another query; this is enough to find an index $i_2 \in Z \setm \{i_1\}$. Repeating this $|Z|$ times, we can reconstruct $Z$ from $A \bz$ alone. (This works for all $Z \ne \varnothing$; since we query $\bone^\txp$, we can also detect when $|Z| = 0$.) By the pigeonhole principle, the number of possible values of $A \bz$ must be at least the number of valid inputs, so $M^{d+1} \ge 2^n$, which implies $d \ge n/\log M - 1$.
\end{proof}

For our round elimination argument, we require the following claim, similar to \Cref{clm:z2-partition} and \Cref{clm:zq-partition}. Even though the claim looks similar, the round elimination argument will be subtly different from its $\ZZ_2$ and $\ZZ_q$ predecessors.

\begin{claim} \label{clm:zz-sunflower}
  Every matrix $A \in \ZZ^{d \times n}$ admits an $A$-uniform family $S_1, \ldots, S_m$ of size $m \ge n / (c_0 d \log n \log M) - 1$, for some absolute constant $c_0$.
\end{claim}

\begin{proof}
  Put $t = \ceil{d \log M}$. Since $\Pi$ is $M$-bounded, the mapping $\bx \mapsto A\bx$ sends the vectors in $\{\bx \in \b^n:\, |\bx| = t\}$ to vectors in $\ZZ^d$ where each entry comes from a set of cardinality $M$. By the pigeonhole principle, there exists a vector $\btr \in \ZZ^d$ such that
  \begin{align} \label{eq:pigeonhole}
    \cF := \{\bx \in \b^n:\, |\bx| = t \text{ and } A\bx = \btr\}
    \text{~~has cardinality~~} |\cF| \ge \dbinom{n}{t} M^{-d} \,.
  \end{align}
  We claim that $\cF$ contains an $m$-sunflower for some integer $m$. Indeed, take $m$ to be the largest integer satisfying
  \begin{align} \label{eq:m-constraint}
    mt \log n < \frac{n}{2c_1} \,, \qquad \text{ which ensures that } \qquad m \ge \frac{n}{c_0 d \log n \log M} - 1 \,.
  \end{align}
  This satisfies the claimed bound upon taking $c_0 = 2c_1$ (say). Continuing from \cref{eq:pigeonhole},
  \begin{alignat*}{2}
    |\cF|
    &\ge \left(\frac{n}{t}\right)^t M^{-d} &\qquad\qquad & \lhd \text{ standard estimate} \\
    &\ge \left(\frac{n}{2t}\right)^t && \lhd \text{ definition of $t$} \\
    &\ge (c_1 m \log n)^t && \lhd \text{ by \cref{eq:m-constraint}} \\
    &\ge (c_1 m \log(mt))^t \,, && \lhd \text{ by \cref{eq:m-constraint}, again}
  \end{alignat*}
  whence the required sunflower exists, by \Cref{thm:sunflower}.
  
  Let $\tS_1, \ldots, \tS_m$ be sets constituting such an $m$-sunflower and let $V = \bigcap_{i=1}^m \tS_i$ be the common pairwise intersection. Define $S_i = \tS_i \setm V$, for each $i \in [m]$. We then have $A\bs_i = A(\bts_i - \bv) = \btr - A\bv$ for each $i$, whence $S_1, \ldots, S_m$ is an $A$-uniform family.
\end{proof}
% side note: what is c1? Looking at Tao's presentation of the sunflower proof, it may need to be roughly 2^10 .

\begin{lemma}[Round elimination lemma] \label{lem:zz-round-elim}
  Let $\Pi$ be a $k$-round $M$-bounded $\ZZ$-LQP for $\elemx_n$, where $k \ge 1$ and $n$ is an integer. Then there exists a $(k-1)$-round $M$-bounded $\ZZ$-LQP $\Upsilon$ for $\elemxquarter_m$, such that
  \begin{thmparts}
    \item $\Upsilon$ shadows $\Pi$ through a homomorphism $\varphi_\Upsilon \colon \Upsilon \to \Pi$; \label{clm:zz-costs}
    \item $m / 4$ is a prime number and
    \begin{align*}
      m \ge \frac{n}{2 c_0 d \log n \log M} - 2 \,,
    \end{align*}
    where $d$ is the cost of the root of $\Pi$, and $c_0$ the constant from \Cref{clm:zz-sunflower}. \label{clm:zz-shrinkage}
  \end{thmparts}
\end{lemma}

\begin{proof}
  Let $A \in \ZZ^{d\times n}$ be the label of the root of $\Pi$. By \Cref{clm:zz-sunflower}, there is an $A$-uniform family $S_1,\ldots,S_{m'}$ of size $m'$. By Bertrand's postulate, there exists a prime number $p$ between between $m'/8$ and $m'/4$; let $m = 4 p$. Let $\bx = A \bs_1$, $\br = (m/4) \bx$ and let $L$ be the lifting matrix defined from $\{S_i\}_{i=1}^{m}$ according to \cref{eq:lift-mtx}. Using \Cref{lem:round-el-gen} on $\Pi$, $L$, and $\br$, we obtain a $(k-1)$-round $\ZZ$-LQP $\Upsilon$ that shadows $\Pi$, and solves $\elemx_m$ on all inputs $W \subseteq [m]$ for which $L \bw \ne \bzero$ and $A L \bw = \br$. While the queries performed by $\Upsilon$ may have larger coefficients than those of $\Pi$, the construction of $\Upsilon$ described in \Cref{sec:framework} only restricts the possible results of each individual linear measurement performed, so $\Upsilon$ is still $M$-bounded. Finally, if $|W| = m/4$, then since $L$ has full rank, $L\bw \ne \bzero$; and furthermore \begin{align*}
      A L \bw = \sum_{i \in W} A \bs_i = |W| \bx = \frac{m}{4} \bx = \br \,.
  \end{align*}
  This implies that $\Upsilon$ gives the correct output for $W \subseteq [m]$ fulfilling the promise of $\elemxquarter_m$.
\end{proof}

The preceding round elimination lemma has a key limitation: it requires a protocol for $\elemx_n$ to create one for $\elemxquarter_m$. Because of this, it is not possible to apply the lemma to its own output, and thereby obtain a $k$-round lower bound.  Say we were to try, and $A$ were the matrix at the root of the protocol $\Pi$ for $\elemxquarter_n$. Then if $A$ contained an all-ones row, \Cref{clm:zz-sunflower} might produce an $A$-uniform family with all set sizes $|S_i|$ equal to some constant $b$ which is not a factor of $n/4$. Then lifting inputs $W$ of size $m/4$ to inputs $Z$ of size $n/4$ would fail, because $n/4 = |Z| = b |W|$ would imply that $b$ divides $n/4$, a contradiction.

With that said, we now use our round elimination lemma in a one-shot fashion to obtain our main result for integer LQPs.

\begin{theorem}[Restatement of \Cref{thm:int-lower-bounds}]
  $\LQ_{\ZZ_{[-B,B]}}^2(\elemx_n) = \Omega(\sqrt{n} / (\log^{3/2}(nB)) )$.
  %Every deterministic $2$-round $\ZZ_{[-B,B]}$-LQP for $\elemx_n$
  %costs $\Omega(\sqrt{n} / (\log^{3/2}(nB)) )$.
\end{theorem}
\begin{proof}
Suppose that $\Pi$ is a deterministic 2-round $\ZZ_{[-B,B]}$-LQP for $\elemx_n$, whose root has cost $d_1$. By \Cref{lem:zz-round-elim}, there is a one round $O(nB)$-bounded protocol for $\elemxquarter_m$ with cost $d_2$. Combining the following three equations:
\begin{alignat*}{2}
    \cost(\Pi) &\ge d_1 + d_2 \\
    d_2 & \ge \frac{0.14 m}{\log M} &\qquad& \lhd \text{ from \Cref{lem:int-quarter-oneround}}\\
    m & \ge \frac{n}{2 c_0 d_1 \log n \log M} - 2 && \lhd \text{ from \Cref{lem:zz-round-elim}}
\end{alignat*}
gives
\[
  \cost(\Pi) \ge 0.19 \sqrt{\frac{n}{c_0 \log n \log^2 M}} - 2 = \Omega\left(\frac{\sqrt{n}}{\log^{3/2}(nB)}\right) \,.
  \qedhere
\]
\end{proof}

\section{Upper Bounds} \label{sec:upper-bounds}

For the sake of completeness, we provide details of the LQPs attaining various upper bounds referenced throughout the paper. For the most part, these upper bounds are simple observations or extensions of well-known existing results.

\subsection{Deterministic \emph{k}-round LQP for \elemx}

The following family of protocols works both when $D = \ZZ_{[0,1]}$ on the problem $\elemx_n$, and when $D = \ZZ_q$ on the problem $\elemxq_n$. The algorithm appears to be well known, and versions of it are described in Lemma 4.1 of \cite{AssadiCK20} and Section 2.2 of \cite{KarpUW88}.

Let $d_1, \ldots, d_{k}$ be a division sequence (see \Cref{lem:division-seq}) for $n$, which minimizes $\sum_{i=1}^{k} d_i$. \Cref{alg:det-protocol} makes no more than $d_r$ queries in each round $r$.

  \begin{algorithm*}[!ht]
  \caption{~Outline of deterministic query protocol on $\bz$}
  \label{alg:det-protocol}
    \begin{algorithmic}[1]
      \State $[u,v] \gets [1,n]$
      \For{$r = 1, \ldots, k $}
        \State Split the interval $[u,v]$ into $d_r+1$ intervals $J_1,\ldots,J_{d_r+1}$, each of size $\le \lceil \frac{v-u+1}{d_r + 1} \rceil$ \label{line:det-proto-split}
        \State Query with matrix $A \in D^{d_r \times n}$, where $A_{i,j}$ is $1$ if $j \in J_i$ and $0$ otherwise.
        \State If $A \bz$ is not all zero, let $i \in [d_r]$ be the index of any nonzero entry; otherwise, let $i = d_r + 1$.
        \State Update $[u,v] \gets J_i$.
      \EndFor
      \State Report $u$ as the index where $u \in Z$.
    \end{algorithmic}
  \end{algorithm*}

Since $d_1, \ldots, d_{k}$ is a division sequence for $n$, the final interval $[u,v]$ must have $u=v$. The total cost of the protocol is $\sum_{i=1}^{k} d_i$, which by \Cref{lem:division-seq} lies in the interval $[k (n^{1/k} - 1), k (\lceil n^{1/k} \rceil - 1)]$. Note that when $n = 2^k$, the algorithm cost is exactly $k$.

Write $\bone_{S}$ to denote the indicator vector in $D^n$ for a given set $S \subseteq [n]$. To prove that the algorithm is correct, it suffices to verify that $\bone^\txp_{[u,v]} \bz \ne 0$ in each round. Since $\bone^\txp \bz \ne 0$, this is true at the start. For any given round, the matrix $A$ queries $\bone_{J_1},\ldots,\bone_{J_{d_r}}$. Since $\bone_{[u,v]} = \sum_{i=1}^{d_r + 1} \bone_{J_i}$, and $\bone^\txp_{[u,v]} \bz \ne 0$, there must be some first index $i$ for which $\bone^\txp_{J_i} \bz \ne 0$. If $i < d_r + 1$, the index is shown in the query response; if $i = d_r + 1$, then no other intervals $J_h$ have $\bone^\txp_{J_h} \bz \ne 0$, so $A \bz$ is all zeros. In either case, the algorithm correctly identifies the interval $J_i$ for which $\bone^\txp_{J_i} \bz \ne 0$.

\subsection{Randomized 1-round LQP for \elemx}

The $\ell_0$-sampling algorithm from \cite{JowhariST11} relies on a standard result on the exact recovery of sparse vectors in $\RR^n$, which (paraphrasing) states that $O(s)$ $\RR$-linear queries suffice to exactly recover any $s$-sparse vector $\bv$ in $\RR^n$, or if $\bv$ is not sparse, say that the output is \textsc{dense} with high probability. The $\ell_0$-sampling algorithm then chooses subsets $\{T_i\}_{i=1}^{\lceil \log n \rceil}$ where each $T_i$ is uniformly randomly drawn from the set of all subsets of $[n]$ of size $2^i$. To obtain a constant final error probability, for each set $T_i$, the $\ell_0$-sampler runs the sparse recovery method on the coordinates given by $T_i$ with $s = O(1)$. The sampler then returns a random index from the first sparse recovery instance to successfully recover a nonzero vector. With high probability, at least one of the sets $T_i$ will contain fewer than $O(1)$ entries of $Z$, and the algorithm succeeds.

Recovering $s$-sparse vectors in $\{0,1\}^n$ is easier than recovering general $s$-sparse vectors in $\RR^{n}$ or $\ZZ^{n}$, so directly adapting \cite{JowhariST11}'s $\ell_0$-sampling algorithm to $\elemx$ means only $O(\log n)$ queries are needed for $\ZZ_{[-B,B]}$ with $B = O(\poly (n))$, and $O(\log ^2 n / \log q)$ for $\ZZ_q$. This follows from the costs of $s$-sparse recovery and detection with $D$-linear queries and $\b^n$, addressed in the following lemma. We spell out this result and its proof for the sake of completeness: though it may be folklore, it appears not to have been published in quite this form.

\begin{lemma}[Discrete $s$-sparse recovery]\label{lem:sparse-recovery}
  There exists a query matrix $H \in \ZZ_{[-B,B]}^{r \times n}$ for $r = O(s \log n / \log B)$ for which the query $H v$ returns a unique value for all $V \subseteq [n]$ with $|V| \le s$. The same holds true for $\ZZ_q$ with $r = O(s \log n / \log q)$.
\end{lemma}
\begin{proof}
  Call a matrix in $A \in D^{r \times t}$ full-$[-1,1]$-rank if there does not exist a nonzero vector $\bv \in \{-1,0,1\}^t$ for which $A \bv = \bzero$. If we choose a matrix $B \in D^{r \times t}$ uniformly at random, then it is full-$[-1,1]$-rank with probability $\ge 1 - 3^t/|D|^r$. One way to prove this is to consider columns the $\bb_1 \ldots \bb_t$ of $B$ one by one, and note that if each $\bb_i$ is not contained in the set $F_i := \{ \sum_{j=1}^{i-1} a_i \bb_i : a \in \{-1,0,1\}^{i-1}\} $, then $B$ has full-$[-1,1]$-rank. Since $B$ is chosen uniformly at random, each column is independent of the the earlier ones, so
  \begin{align*}
     \Pr[ \text{D doesn't have full $[-1,1]$-rank} ]
     & \le \sum_{i=1}^{t} \Pr[\bb_i \notin F_i ]
     \le \sum_{i=1}^{t} \frac{3^{i-1}}{|D|^r}
     \le \frac{3^t}{|D|^r} \,.
  \end{align*}

  Let $r$ be chosen later; if we pick $\hat{H} \in D^{ r \times n}$ uniformly at random, then the expected number of sets $T \subseteq [n]$ with $|T| = 2s$ for which $\hat{H}_T$ (the submatrix of $\hat{H}$ with columns in $T$) has full $[-1,1]$-rank is $\le \binom{n}{2s} 3^{2s} / |D|^n$. Letting $r = \lceil 2 s \log(3n) / \log(|D|) \rceil$ makes this less than $1$. Consequently, there must exist a specific matrix $H$ for which every such submatrix $H_T$ has full $[-1,1]$-rank. Then for any two distinct vectors $\bu,\bv \in \{0,1\}^n$ with $|U|,|V| \le s$, we cannot have $H \bu = H \bv$, because that would imply there exists $T \supseteq U \cup V$ with $|T|=2s$ for which $H_T (bu - bv) = 0$, contradicting the full $[-1,1]$-rank assumption.
\end{proof}

Detecting whether a $\{0,1\}^n$ vector is not $s$-sparse is also easier than in $\RR^n$. For $\ZZ_{[-B,B]}$-LQPs, querying with the vector $\bone \in \ZZ^{d}$ suffices. For $\ZZ_q$, because \Cref{lem:sparse-recovery} ensures that if a vector $\bz$ is s-sparse, it can be recovered exactly, it is enough to query $O(1)$ random vectors in $\ZZ_q^n$. Let $\br$ be such a random vector, and let $\bw$ be the $s$-sparse vector in $\{0,1\}$ recovered using $H$; if $\bz$ was s-sparse, then $\bz = \bw$ and $\br^\txp \bz = \br^\txp \bw$; otherwise, $\br^\txp \bz$ does not equal $\br^\txp \bw$ with probability $1 - 1/q$.

\section{Connections Between $\ZZ_2$-LQPs and Circuit Complexity} \label{sec:kw-lb}

A weaker version of \Cref{thm:z2-k-round} can be proven by combining existing results. As shown in the following lemma, a given $k$-round $\ZZ_2$-LQP $\Pi$ for $\elemxodd$ can be converted to a communication protocol $\Upsilon$ for the Karchmer-Wigderson game on $\textsc{parity}_n$, with the communication cost $C$ of $\Upsilon$ being $\le 2 \cost(\Pi)$. By a slight adaptation of the proof of Theorem 5 in~\cite{NisanW93}, we can convert $\Upsilon$ into an unbounded fan-in boolean formula with depth $k+1$ and no more than $2^C - 1$ AND/OR gates that computes $\textsc{parity}_n$. Relatively tight lower bounds on the size of such a formula date back to \cite{Hastad86}, but we use a result of \cite{Rossman15}, which says that a depth-$(k+1)$ unbounded fan-in formula computing $\textsc{parity}_n$ must have at least $2^{\Omega(k (n^{1/k}-1)}$ AND/OR gates. Thus $\cost(\Pi) \ge \frac{1}{2} C \ge \Omega(k (n^{1/k}-1))$.
%\mscomment{"can convert $\Upsilon$"; in fact, exists correspondence between protocols and formulas. Also, is there a name for boolean formulas with all NOT gates pushed to the leaves?}

\begin{lemma}
Consider the Karchmer-Wigderson game for $\textsc{parity}_n$, in which Alice has a set $X \in \b^n$ with $|X|$ even, and Bob has a set $Y \in \b^n$ with $|Y|$ odd, and they seek to identify an index $i\in[n]$ for which $\bx_i \ne \by_i$. Let $\Pi$ be a $k$-round $\ZZ_2$-LQP for $\elemxodd_n$; then there exists $\Upsilon$ a $k$-round communication protocol for this game, with cost $\le 2 \cost(\Pi)$.
\end{lemma}

\begin{proof}
Let $\rho$ be the root of $\Pi$, with label $A_\rho \in \ZZ_2^{d_\rho \times n}$. In the first round of $\Upsilon$, Alice sends $A_\rho \bx$ to Bob. Then Bob computes $A_\rho \by$, and uses Alice's message to determine $\br_1 = A_\rho (\bx + \by)$. The value $\br_1$ determines a child node $\nu$ of $\rho$. If this is a leaf, Bob outputs its label $o_\nu$. Otherwise, in the second round, Bob sends both $A_\rho \by$ and $A_\nu \by$ to Alice. Given $A_\rho \by$, Alice can determine $\nu$, and compute $A_\nu \bx$. With this, Alice can compute $\br_2 = A_\nu (\bx + \by)$, and identify the child node $\mu$ of $\nu$. If this is a leaf, Alice outputs $o_\mu$; otherwise, in the third round, Alice sends $A_\nu \bx$ and $A_\mu \bx$ to Bob; the players continue in this fashion until a leaf is reached and the protocol ends; since $\Pi$ has depth $k$, this takes at most $k$ rounds.

This protocol is correct, because it finds the leaf of $\Pi$ associated to the input $\bx + \by$. Since we are promised $\bx$ has even parity, and $\by$ odd, $\bx + \by$ has odd parity and thus fulfills the condition under which a protocol for $\elemxodd_n$ must be correct. The output value is an index $i$ where $\bx_i + \by_i = 1$, hence where $\bx_i \ne \by_i$, as required for the communication game.

Since $A_\rho \bx \in \ZZ_2^{d_\rho}$, the round first message uses exactly $d_\rho$ bits. The second, $d_\rho + d_\nu$, the third, $d_\nu + d_\mu$, and so on. The communication needed on inputs ($\bx,\by$) is thus at most twice $\cost(\Pi,\bx+\by)$, so the worst-case communication cost of $\Upsilon$ is at most $2 \cost(\Pi)$.
\end{proof}

\section{Appendix} \label{sec:appendix}

The following estimate was used in \Cref{sec:zq-linear-measurements} during calculations in the proof of our $\ZZ_q$-LQP lower bound.

\begin{lemma}\label{lem:knwk-lb}
Let $C,D$ be constants with $2 C \le D$ and $D \ge 1$. Then
\begin{align}\label{eq:knwk-lb}
    \max\left( \frac{\ln n}{C},k \left(\frac{1}{D} n^{1/k} -1 \right) \right) \ge \frac{1}{D (1 + C)} k \left( n^{1/k} -1 \right) \,.
\end{align}

\begin{proof}
Let $\gamma_n(k) = k(n^{1/k} - 1)$. We have $k \left(\frac{1}{D} n^{1/k} -1 \right) \ge \frac{1}{D} \gamma_n(k) - k$. Since $\gamma_n(k)$ is decreasing, let $k_\star$ be the unique solution to $\frac{1}{D} \gamma_n(k_\star) = \frac{1}{C}\ln n$. Since $\gamma_n(\ln n) = (e-1)\ln n \le 2\ln n \le \frac{D}{C} \ln n$, it follows $k_\star \le \ln n$. Let $k_\dagger$ be the unique solution to $\frac{1}{D} \gamma_n(k_\dagger) - k_\dagger = \frac{1}{C}\ln n$. Since $k_\dagger \le k_\star$, $k_\dagger \le \ln n$ as well. Evaluating the right hand side of \cref{eq:knwk-lb} at $k_\dagger$ gives:
\begin{align*}
    \frac{1}{D (1+C)} \gamma_n(k_\dagger) &= \frac{1}{D (1+C)} \left(\frac{D \ln n}{C} + D k_\dagger \right) \\
            &\le \frac{1}{D (1+C)} \left(\frac{D \ln n}{C} + D \ln n \right) \\
            & = \frac{\ln n}{C} = \frac{1}{D} \gamma_n(k_\dagger) - k_\dagger \,.
\end{align*}
Because the derivative of $\frac{1}{D (1+C)} \gamma_n(k)$ is less that of $\frac{\ln n}{C}$ when $k \ge k_\dagger$, and greater than that of $\frac{1}{D} \gamma_n(k) - k$ for $k \le k_\dagger$, we can extend this inequality to all $k \in (0,\infty)$, proving \cref{eq:knwk-lb}.
\end{proof}

\end{lemma}

\bibliographystyle{alpha}
\bibliography{refs}

\end{document}